%% file: from2024.tex
\documentclass{eptcs}
\usepackage{amsthm}
\input{setup/preamble-from2024.tex}

\input{setup/macros-from2024.tex}
\usepackage{underscore}
\usepackage{xspace}

\begin{document}
\title{A Type System for Data Flow and Alias Analysis in ReScript}

\author{Nicky Ask Lund \institute{Department of Computer Science\\
    Aalborg University, Denmark}\email{loevendallund@gmail.com} \and Hans Hüttel \institute{Department of Computer Science\\
    University of Copenhagen, Denmark}\email{hans.huttel@di.ku.dk}} 

\def\titlerunning{A Type System for Data Flow and Alias Analysis in ReScript}
\def\authorrunning{Lund and Hüttel}

\maketitle


\begin{abstract}
	ReScript is a strongly typed language that targets
        JavaScript, as an alternative to gradually typed languages,
        such as TypeScript. In this paper, we present a sound type system
        for data-flow analysis for a subset of the ReScript language,
        more specifically for a $\lambda$-calculus with mutability and
        pattern matching. The type system is a local analysis that
        collects information about variables that are used at each program
        point as well as alias information. 
\end{abstract}

\section{Introduction}

The goal of data-flow analysis is to provide a static analysis of the
flow information in a program that can be used in compiler
optimizations and for register allocation. The original approach is to
build a system of flow equations based on a graph representation of
the program and to compute a solution using an iterative algorithm
\cite{KildallGaryA1973Auat,RyderBarbara1988Idaa}. Other graph-free
approaches have also been considered
\cite{mohnen}.  A challenge in this setting is how to
deal with the aliasing that imperative language constructs introduce.

Type systems have often been used to provide static analyses of
programs in order to characterize specific run-time errors, including
ones caused by aliasing. In \cite{smith} Smith et
al. present a notion of alias types that allows functions to specify
the shape of a store and to track the flow of pointers through a
computation. The language is a simple location-based language. Other
type systems are substructural. Ahmed et
al. \cite{DBLP:conf/tlca/MorrisettAF05} use a language based on a
linear $\lambda$-calculus to give an alternative formulation of alias
types. The type system crucially relies on linearity, and every
well-typed program terminates.

In this paper we present a type system for data flow analysis in the
presence of aliasing for a non-trivial fragment of the programming
language ReScript which is meant as an alternative to other typed
languages that target JavaScript. ReScript is based on OCaml with a
JavaScript-inspired syntax and a type system based on that of
OCaml\cite{rescript}. ReScript is imperative and allows for
mutability through reference constructs for creation, reading, and
writing.

The fragment that we study incorporates both functional and imperative
features. We show that the type system is sound in the sense that it
correctly overapproximates the set of occurrences on which any given
occurrence depends. Moreover, the dependency information that the type
system provides can also be used to reason about information flow
properties such as non-interference.  Furthermore, an implementation
for the type system has been made to demonstrate the type system. The
full version of our paper is available at \cite{arxiv}.

\section{A fragment of ReScript}\label{sec:lang}

In this paper we consider a fragment of ReScript that contains a $\lambda$-calculus
with pattern matching, local declarations that can be recursive and a notion of mutable references.
 
\subsection{Syntax}

In a data-flow analysis we must record information of where variables
are used.  Therefore, the language presented is extended with a notion
of \emph{program points} taken from a countably infinite set
$\cat{P}$.
Every subexpression is labelled with a unique program
point. \emph{Occurrences} $o \in \cat{Occ}$ are labelled expressions
$e^p$ where $e \in \cat{Exp}$ and $p \in \cat{P}$. We let $\ell$ range
over a countably infinite, totally ordered set of locations
$\cat{Loc}$ and $x,f$ range over the set of variables $\cat{Var}$.  An
occurrence is \emph{atomic} if it is of the form $u^p$ where
$u \in \cat{Var} \cup \cat{Loc}$. If $u^p$ is atomic, we call any
other occurrence $u^q$ a \emph{$u$-occurrence}.

When given a syntactic category $\mathbf{C}$, we let
$\mathbf{C}_{\cat{P}}$ denote the pair $\mathbf{C}\times\cat{P}$, so
that e.g.  $\cat{Exp}_\cat{P}=\cat{Exp}\times\cat{P}$.  This means
that $\cat{Occ}=\cat{Exp}_{\cat{P}}$.

The formation rules of our abstract syntax are shown below.
\begin{align*}
o &::= e^p \\
e &::= x \mid c \mid o_1\;o_2 \mid \lambda x.o \mid c \; o_1 \; o_2\\
			& \mid \mbox{let} \; f \; o_1 \; o_2 \mid
                   \mbox{let rec} \; f \; o_1 \; o_2 \mid \mbox{case}
                   \; o_1 \; \vec{\pi} \; \vec{o} \mid  \refc \; o \mid o_1 := o_2 \mid \; !o\\
\pi &::= n \mid b \mid x \mid \_  \mid
                                (s_1,\cdots,s_n)
\end{align*}
An \emph{abstraction} $\lambda\;x.o$ has a parameter $x$ and
body $o$.  \emph{Constants} $c$ are either natural numbers $n$,
boolean values $b$, unit value $()$, or functional constants, such as
the arithmetic operations and the Boolean connectives.

An \emph{application} is written $o_1\;o_2$ , and $c\;o_1\;o_2$ denotes a
\emph{functional application} where $c$ is a functional constant and
$o_1$, $o_2$ are its arguments.

\emph{Local declarations} $\mbox{let} \; f \; o_1 \; o_2$ associate
the variable $f$ with the value $o_1$ within $o_2$, and
$\mbox{let rec} \; f \; o_1 \; o_2$ allows us to define a recursive
function $f$ where $f$ may occur in $o_1$. ReScript is an imperative
language due to the presence of the \emph{reference} construct
$\mbox{ref\;o}$ which creates a reference in the form of a location
and allows for binding locations to local declarations. We can read
from a reference $o$ by writing $!o$ and write to a reference using
the \emph{assignment} construct $o_1\;:=\;o_2$.

The pattern matching construct
$\mbox{case} \; o_1 \; \vec{\pi} \; \vec{o}$ matches an occurrence
with the ordered set, $\vec{\pi}$, of patterns.  We denote the size of
the tuple pattern $\pi$ by $|\vec{\pi}|$ and the size of a tuple
occurrence by $|\vec{o}|$, requiring that $|\vec{\pi}| = |\vec{o}|$
such that for each pattern in $\vec{\pi}$ there is a clause in
$\vec{o}$.

The notions of free and bound variables are defined as expected. We
assume that all binding occurrences involve distinct names; this can
be ensured by means of $\alpha$-conversion.

\begin{example}\label{ex:write}
Consider
 \begin{lstlisting}[language=Caml, mathescape=true]
   (let x (ref 484000$\p{1}$)$\p{2}$
     (let y (let z (5$\p{3}$)$\p{4}$
        (x$\p{5}$:=z$\p{7}$)$\p{8}$)$\p{9}$ (!x)$\p{10}$)$\p{11}$)$\p{12}$
\end{lstlisting}
This creates a reference to the constant 3 and binds the reference to
$x$ (so $x$ is an alias of this reference). Next a binding of $z$ is made
to the constant 5 before writing to the reference, that $x$ is bound to, 
to the value that $z$ is bound to. Then a binding for $y$ is made to the unit value,
as the assignment evaluates to the unit value. Lastly the reference, that $x$ is bound to, is read.

\end{example}

\subsection{The binding model}\label{sec:EnvSto}

Our binding model uses an environment $env$ that keeps
track of the bindings of variables to values. Values are given by the formation rules
\[ v ::= c \mid \ell \mid () \mid \langle x,e^{p'},env\rangle \mid \langle
  x,f,e^{p''},env\rangle \]
Constants $c$, locations $\ell$, and unit $()$ are values, as are closures, $\langle
x,e^{p'},env\rangle$ and recursive closures, $\langle x,f,e^{p''},env\rangle$.

An environment $env \in \cat{Env} $ is a partial function $env:
\cat{Var}\rightharpoonup\cat{Values}$ and we let $env^{-1}(v)=\{x\in \dom(env)\mid env(x)=v\}$.
A store $sto \in \cat{Sto}$ is a partial function
$\cat{Sto}=\cat{Loc}\cup\{\nexte\}\rightharpoonup\cat{Values}$ where
$\nexte$ is a pointer to the next unused location -- this information
is needed when new locations are needed.

Moreover, we assume a function
$new:\cat{Loc}\rightarrow\cat{Loc}$, which given, a location, gives us
the next location.

For any function $f$ we let $f[u
\mapsto w]$ denote the function $f'$ such that $f(u') = f(u)$ for $u
\neq u'$ and $f'(u') = w$.

\subsection{Keeping track of dependencies}\label{sec:DepFunc}

The semantics that follows will collect the semantic dependencies in a
computation. An occurrence $u^p$ semantically depends upon a set of
occurrences $S$ if the value of $u^p$ can be found using at most the
values of the occurrences in $S$. To determine semantic dependencies
we use a dependency function that will tell us for each variable and
location occurrence what other, previous occurrences they depend upon.

\begin{definition}[Dependency function]\label{def:DepFunc}
  A dependency function $w$ is a partial functions from atomic
  occurrences to a pair of dependencies:
\[
  w: \cat{Loc}_{\mathbf{P}}\cup\cat{Var}_{\mathbf{P}}\rightharpoonup\Pow{\cat{Loc}_{\mathbf{P}}}\times\Pow{\cat{Var}_{\mathbf{P}}} \] 
\end{definition}
For a dependency function $w$ and a
$u^p\in\cat{Loc}_{\mathbf{P}}\cup\cat{Var}_{\mathbf{P}}$, the clause
\[ w (u^p)=(L,V) \]
tells us that the element $u^p$ is bound to a pair of location and
variable occurrences where $L$ is a set of location occurrences
$L=\{\loc_1^{p_1},\cdots,\loc_n^{p_n}\}$ and a set of variable
occurrences $V =\{x_1^{p'_1},\cdots,x_m^{p'_m}\}$, meaning that the
value of the element $u^p$ depends on the occurrences found in $L$ and
in $V$.

\begin{example}\label{ex:dep}
	Consider the occurrence from \cref{ex:write}, where we can
        infer the following bindings for a dependency function
        $w_{ex}$ over this occurrence: 
	\[
          w_{ex}= \begin{array}{l}
                    [x^2\mapsto(\emptyset,\emptyset),z^4\mapsto(\emptyset,\emptyset),y^9\mapsto(\emptyset,\{x^5\}),
                    \\ \; \loc^2\mapsto(\emptyset,\emptyset),\loc^8\mapsto(\emptyset,\{z^7\})] \end{array} \]
	where $\loc$ is the location created from the reference construct.
	The variable bindings are distinct, as the location $\loc$ is
        bound multiple times, for the program points $2$ and $8$. 

        When we read an occurrence bound in $w_{ex}$, we must also
        know its program point, as there can exists multiple bindings
        for the same variable or location.
\end{example}

By considering \cref{ex:dep}, we would like to read the information
from the location, that $x$ is an alias to.  As it is visible from the
occurrence in \cref{ex:write}, we know that we should read from
$\loc^8$, since we wrote to that reference at the program point $8$.
From $w_{ex}$ alone it is not possible to know which occurrence to
read, since there is no order defined between the bindings.  We therefore
introduce a notion of ordering in the form of a binary relation over
program points.



\begin{definition}[Occurring program points]\label{def:OccPP}
	Let $O$ be a set of occurrences, then $\points(O)$ is given by:
	\[ \points(O)=\{p\in\cat{P}\mid\exists e^p. e^p\in O\} \]
        For a pair $(L,V)$ we let $\points(L,V) = \points(L) \cup \points(V)$.
\end{definition}

Any dependency function induces an ordering on program points as follows.

\begin{definition}\label{def:RelPoint}
	Let $w\in\cat{W}$ be a dependency function.
	Then the induced order $\sqleq_w$ is given by
	 \[ \sqleq_w= \{(p,p')\mid p \in \points(\dom(w)), p' \in
           \points(w(p)) \} \]
         	We say that $w$ is a partial order if its equality
                closure $\sqleq_w$ is a partial order.
\end{definition}

\begin{example}\label{ex:depRel}
	Consider the example from \cref{ex:dep}. Assume a binary
        relation over the dependency function $w_{ex}$ given by
	\[ \sqleq_{w_{ex}}=\{(2,4),(2,9),(5,9),(2,8),(7,8)\} \]
	From this ordering, it is easy to see the ordering of the elements.
	The ordering we present also respects the flow the occurrence from \cref{ex:write} would evaluate to.
	We then know that the dependencies for the reference (that $x$ is an alias to) is for the largest binding of $\loc$.
\end{example}

%
The immediate predecessor $\IP(u,S)$ of an element $u^p$ wrt. a set of
occurrences $S$ is the most recent $u$-element in $S$ seen before $u$. 

\begin{definition}[Immediate predecessor]\label{def:GBind}
  Let $u$ be an element, let \sqleq\ be an ordering on program points
  and $S$ be a set of occurrences, then $\IP(u,S)$ is given by
\[ \IP(u^p,S)=\sup\{u^q\in S\mid q\sqleq p\} \]
\end{definition}

Based on \cref{def:GBind}, we can present an instantiation of the
function for the dependency function $w$ and an order over $w$,
$\sqleq_w$: 

\begin{definition}
  Let $w$ be a dependency function, $\sqleq_w$ be the induced order,
  and $u$ be an element, then $\IP_{\sqleq_w}$ is given by:
	$$\IP_{\sqleq_w}(u,w)=\sup\{u^p\in \dom(w)\mid u^q\in \dom(w).q\sqleq_w p\}$$
\end{definition}

\begin{example}\label{ex:deplookup}
  As a continuation of \cref{ex:depRel}, we can now find the greatest
  element for an element, e.g., a variable or location.  As we were
  interested in finding the greatest bindings a location is bound to
  in $w_ex$, we use the function $\IP_{\sqleq_w}$:
	\[ \IP_{\sqleq_{w_ex}}(\loc,w_ex)=\sup\{\loc^p\in \dom(w)\mid
          \loc^q\in \dom(w). q\sqleq_{w_ex} p\} \]
where the set we get for $\loc$ is 
$\{\loc^2,\loc^8\}$. From this, we find the greatest element: 
	\[ \loc^8=\sup\{\loc^2,\loc^8\} \]
	As we can see, from the $\IP_{w_ex}$ function, we got $\loc^8$ which were the occurrence we wanted.
\end{example}

\subsection{Collecting semantics}\label{sec:sem}

The semantics for our language that collects dependency information is
a big-step semantics with transitions of the form
\begin{align*}
env\vdash\left\langle e^{p'},sto,(w,\sqleq_w),p\right\rangle\rightarrow\left\langle v,sto',(w',\sqleq_w'),(L,V),p''\right\rangle
\end{align*}
This should be read as: Given the store $sto$, a dependency function
$w$, a relation over $w$, and the previous program point $p$, the
occurrence $e^{p'}$ evaluates to a value $v$, an updated store $sto'$,
an updated dependency function $w'$, a relation over $w'$, the
dependency pair $(L,V)$, and the program point $p''$ reached after
evaluating $e^{p'}$, given the bindings in the environment $env$. 

A selection of the rules for $\rightarrow$ can be found in
\cref{fig:ColSem}.

The \runa{Var} rule  uses the environment to get the value $x$ is bound to and uses dependency function $w$ to get its dependencies.
		To lookup the dependencies, the function $\IP_{\sqleq_w}$ is used to get the greatest binding a variable is bound to, in respect to the ordering $\sqleq_w$.
		Since the occurrence of $x$ is used, it is added to the set of variable occurrences we got from the lookup of the dependencies for $x$.

The \runa{Let} rule for the occurrence $[\mbox{let}\;x\;e_1^{p_1}\;e_2^{p_2}]^{p'}$, creates a local binding that can be used in $e_2^{p_2}$.
		The \runa{Let} rule evaluate $e_1^{p_1}$, to get the value $v$, that $x$ will be bound to in the environment for $e_2^{p_2}$, and the dependencies used to evaluate $e_1^{p_1}$ are bound in the dependency function.
		As we reach the program point $p_1$ after evaluating $e_1^{p_1}$, and it is also the program point before evaluating $e_2^{p_2}$, the binding of $x$ in $w$ is to the program points $p_1$.	

The \runa{Ref} rule, for the occurrence $[\refc\;e^{p'}]^{p''}$,
creates a new location and binds it in the store $sto$, to the value
evaluated from $e^{p'}$. We record the dependencies from evaluating
the body $e^{p'}$ in $w$ at the program point $p''$.

\runa{Ref-read} evaluates the body $e^{p_1}$ to a value which
must be a location $\loc$, and reads the value of $\loc$ in the store.
		The \runa{Ref-read} rule looks up the dependencies for
                $\loc$ in $w$.
		As there could be multiple bindings for $\loc$, in $w$, at different program points, we use the $\IP_{\sqleq_{w'}}$ function to get greatest binding of $\loc$ with respect to the ordering $\sqleq_{w'}$, 
		and we also add the location occurrence $\loc^{p'}$ to the set of locations.

                Finally \runa{Ref-write} tells us that we must
                evaluate $e_1^{p_1}$ to a location $\loc$ and
                $e_2^{p_2}$ to a value $v$, and bind $\loc$ in the
                store $sto$ to the value $v$.  We pass the program
                point $p'$ and the dependency function
                is also updated with a new binding for $\loc$.

\begin{table*}
	\begin{tabular}{l}
          \runa{Var}\\[0.2cm]
	\inference[]{}
	{env\vdash \left\langle x^{p'},sto,(w,\sqleq_w),p \right\rangle \rightarrow \left\langle v,sto,(w,\sqleq_w),(L,V\cup\{x^{p'}\}),p' \right\rangle}\\[6mm]
	where $env(x)=v$, $x^{p''}=\IP_{\sqleq_w}(x,w)$, and
          $w(x^{p''})=(L,V)$ \\[4mm]
                    \runa{Let}\\[0.2cm]
	\inference[]
	{
		env\vdash \left\langle e_1^{p_1},sto,(w,\sqleq_w),p \right\rangle \rightarrow \left\langle v_1,sto_1,(w_1,\sqleq_w^1),(L_1,V_1),p_1 \right\rangle &\\
		env[x\mapsto v_1]\vdash \left\langle e_2^{p_2},sto_1,(w_2,\sqleq_w^1),p_1 \right\rangle \rightarrow \left\langle v,sto',(w',\sqleq_w'),(L,V),p_2 \right\rangle
	}
	{env\vdash \left\langle \left[\mbox{let}\;x\;e_1^{p_1}\;e_2^{p_2}\right]^{p'},sto,(w,\sqleq_w),p \right\rangle \rightarrow \left\langle v,sto',(w',\sqleq_w'),(L,V),p' \right\rangle}\\[6mm]
	where $w_2=w_1[x^{p_1}\mapsto(L,V)]$
          \\[4mm]
          \runa{Ref}\\
	\inference[]
	{env \vdash \left\langle e^{p'},sto,(w,\sqleq_w),p \right\rangle \rightarrow \left\langle v,sto',(w',\sqleq_w'),(L,V),p' \right\rangle}
	{env\vdash \left\langle \left[\refc\;e^{p'}\right]^{p''},sto,(w,\sqleq_w),p \right\rangle \rightarrow \left\langle \loc,sto'',(w'',\sqleq_w'),(\emptyset,\emptyset),p'' \right\rangle}\\[6mm]
	where $\loc=next$, $sto''=sto'[next\mapsto new(\loc),\loc\mapsto v]$, and\\
          $w''=w'[\loc^{p'}\mapsto (L,V)]$ \\[4mm]
 \runa{Ref-read}\\
	\inference[]
	{env \vdash \left\langle e^{p_1},sto,(w,\sqleq_w),p \right\rangle \rightarrow \left\langle \loc,sto',(w',\sqleq_w'),(L_1,V_1),p_1 \right\rangle}
	{env\vdash \left\langle \left[!e^{p_1}\right]^{p'},sto,(w,\sqleq_w),p \right\rangle \rightarrow \left\langle v,sto',(w',\sqleq_w'),(L\cup L_1\cup\{\loc^{p''}\},V\cup V_1),p' \right\rangle}\\[6mm]
	where $sto'(\loc)=v$,
          $\loc^{p''}=\IP_{\sqleq_w'}(\loc,w')$, and
          $w'(\loc^{p''})=(L,V)$ \\[4mm]
\runa{Ref-write}\\
	\inference[]
	{
		env \vdash \left\langle e_1^{p_1},sto,(w,\sqleq_w),p \right\rangle \rightarrow \left\langle \loc,sto_1,(w_1,\sqleq_w^1),(L_1,V_1),p_1 \right\rangle &\\
		env \vdash \left\langle e_2^{p_2},sto_1,(w_1,\sqleq_w^1),p_1 \right\rangle \rightarrow \left\langle v,sto_2,(w_2,\sqleq_w^2),(L_2,V_2),p_2 \right\rangle
	}
	{env\vdash \left\langle \left[e_1^{p_1}:=e_2^{p_2}\right]^{p'},sto,(w,\sqleq_w),p \right\rangle \rightarrow \left\langle (),sto',(w',\sqleq_w'),(L_1,V_1),p' \right\rangle}\\[6mm]
	where $sto'=sto_2[\loc\mapsto v]$,
          $\loc^{p'}=inf_{\sqleq_w^2} \loc,w$,\\
	$w'=w_2[\loc^{p'}\mapsto(L_2,V_2)]$, and
          $\sqleq_w'=\sqleq_w^2\cup(p'',p')$ \\[6mm]
	\end{tabular}
	\caption{Selected rules from the semantics}
	\label{fig:ColSem}
\end{table*}

\section{A type system for data-flow analysis}\label{sec:TypeSys}

The type system for data-flow analysis that we now present is an
overapproximation of the big-step semantics.

\subsection{An overview of the type system}

The system assigns types, presented in \cref{sec:types}, to
occurrences given a type environment (presented in
\cref{sec:Judge}) and a so-called basis (presented in
\cref{sec:basis}).

As presented, the language contains local information as bindings and
global information as locations.  Since locations are a semantic
notion, and references do not need to be bound to variables, we use
the notion of \emph{internal} variables to represent locations.
Internal variables are denoted by
$\nu x,\nu y, \nu z \ldots \in\cat{IVar}$.  We use a partition of
$\cat{IVar} \cup \cat{Var}$ to represent aliasing.  Whenever variables
or internal variables belong to the same subset in a partition, the
intention is that they share the same location.

In this paper, we will not introduce polymorphism into the type
system.  For this reason we require that references cannot be bound to
abstractions and that every abstraction is used at most used once.

\subsection{Types}\label{sec:types}

The set of types \cat{Types} is defined by the formation rules
\[ T ::=(\delta,\kappa)\mid T_1 \rightarrow T_2 \]
If an occurrence $o$ has the base type $(\delta,\kappa)$, the set
$\delta$ is the set of occurrences that the value of $o$ can depend
on, and $\kappa$ represents alias information in the form of the set
of variables and internal variables upon which the value of $o$ may
depend.  If $o$ has type $\kappa \neq \emptyset$, the
occurrence must therefore represent a location.

\begin{definition}[Type base for aliasing]
	For an occurrence $o$, let $\cat{Var}_o$ be the set of all
        variables found in $o$ and $\cat{IVar}_o$ be the set of all
        internal variables found in $o$.
	The type base $\kappa^0=\{\kappa^0_1,\cdots,\kappa^0_n\}$ is
        then a partition of $\cat{Var}_o \cup \cat{IVar}_o$, where
        $\kappa_i^0\cap\kappa_j^0=\emptyset$ for all $i\neq j$. 
\end{definition}
For a variable, $x$ to be an alias of an internal variable, $\nu y$,
there must exist a $\kappa^0_i$ where $x \in \kappa^0_i$ and $\nu y
\in \kappa^0_i$. This means that there can only be more than one
variable in a $\kappa^0_i$, if there also exists an internal variable
in $\kappa^0_i$.


The arrow type is introduced to type abstractions.
If either $T_1$ or $T_2$ in an arrow type $T_1 \to T_2$ is a base type where $\kappa \neq \emptyset$,
then the abstraction must either take a reference as input or return a reference.



Since the type system approximates the occurrences used to evaluate an occurrence, we need a notion of combining types.

\begin{definition}
  	Let $T_1$ and $T_2$ be two types, then their union is defined as

  \[
T_1\cup T_2=
    \begin{cases}
      (\delta\cup\delta',\kappa\cup\kappa') & \begin{array}{l}
                                                T_1=(\delta,\kappa) \\
      T_2=(\delta',\kappa') \end{array} \\[3mm]
       (T_1'\cup T_2')\rightarrow (T_1''\cup T_2'') &
       \begin{array}{l} T_1=T_1'\rightarrow T_1'' \\;T_2=T_2'\rightarrow
       T_2''  \end{array}\\[3mm]
       (\delta' \cup \delta, \kappa \cup \kappa') & \begin{array}{l}
                                                      T_1 =
                                                      (\delta',\kappa') \\
                                                      T_2 =
                                                      (\delta,\kappa) \end{array}
                                                    \\[3mm]
       \text{undefined} & \text{otherwise}
    \end{cases}
\]
\end{definition}

\subsection{The binding model of the type system}\label{sec:basis}

The semantics will let us find dependency information, and the type 
system must approximate these semantic notions.





A type environment tells us the types of elements.

\begin{definition}[Type Environment]
  A type environment $\Gamma$ is a partial function $\Gamma:\cat{Var}_{\mathbf{P}}\cup\cat{IVar}_{\mathbf{P}}\rightharpoonup\cat{Types}$
\end{definition}

\begin{definition}[Updating a type environment]
	Let $\Gamma$ be a type environment, let $u^p$ be an element
        and let $T$ be a type.
	We write $\Gamma[u^p:T]$ to denote the type environment $\Gamma'$ where:
	\begin{align*}
		\Gamma'(y^{p'})=
		\left\{\begin{matrix}
			\Gamma(y^{p'}) & \mbox{if}\;y^{p'}\neq u^{p}\\\	 
			T & \mbox{if}\;y^{p'}=u^{p}
		\end{matrix}\right.
	\end{align*}
\end{definition}

We also assume an ordering of program points at type level.

\begin{definition}[Approximated order of program points]
	An approximated order of program points $\Pi$ is a pair
	\[ \Pi=(\cat{P},\sqleq_\Pi) \]
	where
	\begin{itemize}
		\item \cat{P} is the set of program points in an occurrence,
		\item $\sqleq_\Pi\subseteq\cat{P}\times\cat{P}$
	\end{itemize}
	We say that $\Pi$ is a partial order if $\sqleq_\Pi$ is a partial order.
\end{definition}

The notion of the immediate predecessor of a $u\in\cat{IVar} \cup \cat{Var}$ at the type level is relative to a type environment $\Gamma$ and the approximated order $\Pi$.

\begin{definition}[Immediate predecessor at type level]\label{def:GBindPi}
	\[ \IP_{\sqleq_\Pi}(u,\Gamma)=\sup\{u^p\in
          \dom(\Gamma)\mid u^q\in \dom(\Gamma).q\sqleq_\Pi p\} \]
\end{definition}

A lookup of variable $u^p$ in the semantics is straightforward as its
value will be unique. In the type system, however, we need to
approximate over all possible branches in an occurrence.  To this end,
we consider chains wrt. our approximate order. A $p$-chain describes
the history, or a single possible evalution, behind an occurrence $u^p$.
A set of $p$-chains can thus be used to describe what an internal variable depends on.





\begin{definition}[$p$-chains]
  A $p$-chain, denoted as $\Pi_p^{*}$, is a maximal chain
  wrt. $\sqleq_\Pi$ whose maximal element is $p$. We write
  $\Pi_p^{*}\in\Pi$, if $\Pi_p^{*}$ is a $\Pi$-chain. For any $p$, we
  let $\Upsilon_p$ denote the set of all $p$-chains in $\Pi$.
\end{definition}


We can now define the immediate predecesor for the type system wrt. the set of $p$-chains.
This is done by taking the union of all $p$-chains for an occurrence $u^p$.

\begin{definition}\label{def:GBindUps}
	Let $u\in \cat{Var}\cup\cat{IVar}$, be either a variable or internal variable, $\Gamma$ be a type environment, and $\Upsilon_p$ be a set of $p$-chains, then $\IP_{\Upsilon_p}$ is given by:
\[
  \IP_{\Upsilon_p}(u,\Gamma)=\bigcup_{\Pi_p^{*}\in\Upsilon_p}\IP_{\Pi_p^{*}}(u,\Gamma) \]
\end{definition}

\subsection{The type system}\label{sec:Judge}
We will now present the judgement and type rules for the language, that is, how we assign types to occurrences.

Type judgements have the format
\[ \Gamma,\Pi\vdash e^p: T \]
and should be read as: the occurrence $e^p$ has type $T$, given the
dependency bindings $\Gamma$ and the approximated order of program
points $\Pi$. 

A highlight of type rules can be found in \cref{fig:TypeSys}.

\begin{description}

	\item[\runa{T-Var}] rule, for occurrence $x^p$, looks up the type for $x$ in the type environment, by finding the greatest binding using \cref{def:GBindPi}, and adding the occurrence $x^p$ to the type.

	\item[\runa{T-Let-1}] rule, for occurrence $[\mbox{let}\;x\;e_1^{p_1}\;e_2^{p_2}]^p$, creates a local binding for an internal variable, with the type of $e_1^{p_1}$ that can be used in $e_2^{p_2}$.
     As such, the rule assumes that the type of $e_1^{p_1}$ is a base type with alias information, i.e., $\kappa\neq\emptyset$.
		The other cases, when $e_1^{p_1}$ is not a base type with alias information, are handled by the \runa{T-Let-2} rule.

	\item[\runa{T-Case}] rule, for occurrence $[\mbox{case}\;e^{p}\;\vec{\pi}\;\vec{o}]^{p'}$, is an over-approximation of all cases in the pattern matching expression, by taking an union of the type of each case.
		Since the type of $e^p$ is used to evaluate the pattern matching, we also add this type to the type of the pattern matching.

	\item[\runa{T-Ref-read}] rule, for occurrence $[!e^{p}]^{p'}$, is used to retrieve the type of references, where $e^p$ must be a base type with alias information.
     Since the language contains pattern matching, there can be multiple internal variables in $\kappa$ and multiple occurrences to read from.
     To do the lookup, we use $\IP_{\Upsilon_{p'}}$ to look up all the $p'$-chains.
\end{description}

\begin{table*}[h]
	\setlength\tabcolsep{8pt}
	\begin{tabular}{ll}
\runa{T-Var} &
	\condinf{}
	{\Gamma,\Pi \vdash x^p:T \sqcup (\{x^p\},\emptyset)}{where 
	$x^{p'}=uf_{ \sqleq_\Pi}(x,\Gamma)$, and
          $\Gamma(x^{p'})=T$} \\[15mm]
\runa{T-Let-1} &
	\condinf{
		\Gamma,\Pi\vdash e_1^{p_1}:(\delta,\kappa) \\
		\Gamma',\Pi\vdash e_2^{p_2}:T_2
	}
	{\Gamma,\Pi\vdash [\mbox{let}\; x \; e_1^{p_1} \; e_2^{p_2}]^{p}:T_2}{where $\Gamma'=\Gamma[x^{p}:(\delta,\kappa\cup \{x\})]$ and
          $\kappa\neq\emptyset$} \\
\runa{T-Case} &
	\condinf
	{
		\Gamma,\Pi\vdash e^{p}:(\delta,\kappa) \\
		\Gamma',\Pi\vdash e_i^{p_i}:T_i\;\;\;(1\leq i\leq|\vec{\pi}|)
	}
	{\Gamma,\Pi\vdash [\mbox{case}\;e^{p}\;\vec{\pi}\;\vec{o}]^{p'}:T\sqcup(\delta,\kappa)}
	{where $e_i^{p_i}\in\vec{o}$ and $s_i\in\vec{\pi}$ $T=\bigcup_{i=1}^{|\vec{\pi}|}T_i$, and\\
          $\Gamma'=\Gamma[x^p:(\delta,\kappa)]$ if $s_i=x$} \\[18mm]
\runa{T-Ref-read} &
	\condinf
	{\Gamma,\Pi\vdash  e^{p}:(\delta,\kappa)}
	{\Gamma,\Pi\vdash
                    [!e^{p}]^{p'}:T\cup(\delta\cup\delta',\emptyset)}
	{where $\left\{\begin{array}{l}\kappa\neq\emptyset$, $\delta'=\{\nu x^{p'}\mid\nu x\in\kappa\}$, $\nu x_1,\cdots,\nu x_n\in\kappa\\ 
	\{\nu x_1^{p_1},\cdots,\nu x_1^{p_m}\}=uf_{\Upsilon_{p'}}(\nu
                  x_1,\Gamma),\cdots,\\ \quad\{\nu x_n^{p_1'},\cdots,\nu
                  x_n^{p_s'}\}=uf_{\Upsilon_{p'}}(\nu x_n,\Gamma) \\
	T=\Gamma(\nu x_1^{p_1})\cup\cdots\cup\Gamma(\nu
                         x_1^{p_m})\cup\cdots\cup\\ \quad\Gamma(\nu
                         x_n^{p_1'})\cup\cdots\cup\Gamma(\nu
                         x_n^{p_s'})\end{array}\right\}$}
          \\
          & 
	\end{tabular}
	\caption{Selected rules from the type system}
	\label{fig:TypeSys}
\end{table*}


\section{Soundness}\label{sec:Soundness}

The type system is sound in that the type of
an occurrence correspond to the dependencies and the alias information
from the semantics. 
To show this, we will first introduce the
type rules for values and then describe  relation between the semantics and the
type system. 

\subsection{Type rules for values}

In our soundness theorem and its proof, values are mentioned. We
therefore state a collection of type rules for the values presented
int \cref{sec:EnvSto}. The type rules are given in
\cref{fig:ValTypeRules}. We describe the central ones here.

\begin{description}
\item[\runa{Constant}] differs from the rule \runa{T-Const}, since
  most occurrences can evaluate to a constant and as such we know that
  its type should be a base type. Constants can depend on other
  occurrences; we know that $\delta$ can be non-empty, but since
  constants are not locations, we also know that it cannot contain
  alias information, and as such $\kappa$ should be empty.

\item[\runa{Location}] types locations, and their type must be a base
  type. Since locations can depend on other occurrences, we know that
  $\delta$ can be non-empty.  As locations can contains alias
  information, and that a location is considered to always be an alias
  to itself, we know that $\kappa$ can never be empty, as it should
  always contain an internal variable.

	\item[\runa{Closure}] type rule represents abstraction, and as
          such we know that it should have the abstraction type,
          $T_1\rightarrow T_2$, where we need to make an assumption
          about the argument type $T_1$. 
		Since a closure contains the parameter, body, and the
                environment for an abstraction from when it were
                declared, we also need to handle those part in the
                type rule.
                
		The components of the closure are handled in the
                premises, where the environment must be well-typed.
                We also type the body of the abstraction in a type
                environment updated with the type $T_1$ of its
                parameter.


\end{description}

As closures and recursive closures contain an environment, we also
need to define what it means to be a well-typed environment $env$
wrt. a type environment: Every variable bound in $env$ is bound to a
value that is well-typed wrt. $\Gamma$.

\begin{definition}[Well-typed environments]\label{def:TEnv}
	Let $v_1,\cdots,v_n$ be values such that $\Gamma,\Pi\vdash v_i:T_i$, for $1\leq i\leq n$.
	Let $env$ be an environment given by $env=[x_1\mapsto
        v_1,\cdots,x_n\mapsto v_n]$, $\Gamma$ be a type environment,
        and $\Pi$ be the approximated order of program points. 
	We say that:
	$$\Gamma,\Pi\vdash env$$
	iff 
	\begin{itemize}
		\item For all $x_i\in \dom(env)$ then $\exists x_i^p\in \dom(\Gamma)$ where $\Gamma(x_i^p)=T_i$ then 
			$$\Gamma,\Pi\vdash env(x_i):T_i$$
	\end{itemize}
\end{definition}

\begin{table*}
	\setlength\tabcolsep{8pt}
	\begin{tabular}{ll}
		\runa{Constant} &
			\inference[]{}
				{\Gamma,\Pi\vdash  c:(\delta, \emptyset)}\\[1cm]

		\runa{Location} &
			\inference[]{}
				{\Gamma,\Pi\vdash  \loc:(\delta, \kappa)}\\
				Where $\kappa\neq\emptyset$\\[1cm]

		\runa{Closure} &
			\inference[]
				{
					\Gamma,\Pi\vdash env \\
					\Gamma[x^{p}:T_1],\Pi\vdash e^{p'}:T_2
				}
				{\Gamma,\Pi\vdash \left\langle x^{p}, e^{p'}, env \right\rangle^{p''}:T_1\rightarrow T_2}


	\end{tabular}
	\caption{Type rules for values}
	\label{fig:ValTypeRules}
\end{table*}

\subsection{Notions of agreement}

As our soundness theorem relates the type system to the semantics, we
must define what it means for instances of the binding models of the
semantics and the type system to agree.

First we define what it means for a set of occurrences $\delta$ to
faithfully represent the information from a dependency pair $(L,V)$
wrt. a environment $env$.

\begin{definition}[Dependency agreement]\label{def:DepAgree}
	We say that:
	$$(env,(L,V))\models\delta$$
	if
	\begin{itemize}
		\item $V\subseteq\delta$,
		\item For all $\loc^p\in L$ where
                  $env^{\loc}\neq\emptyset$, we then have $env^{-1} \loc \subseteq \kappa_i^0$ for some $\kappa_i^0\in\delta$
		\item For all $\loc^p\in L$ where $env^{\loc}=\emptyset$ then there exists a $\kappa_i^0\in\delta$ such that $\kappa_i^0\subseteq\cat{IVar}$
	\end{itemize}
      \end{definition}

      Since types can contain alias information $\kappa$, we also need
      to define what it means for the information in $\kappa$ to be
      known to an environment $env$.  If there exists alias
      information in $env$, then there exists an alias base
      $\kappa^0_i\in\kappa^0$ such that the alias information known to
      $env$ is included in that of $\kappa^0_i$, and there exists a
      $\nu x\in\kappa$, such that $\nu x\in \kappa^0_i$.  If there is
      no currently known alias information, we simply check that there
      exists a corresponding internal variable, that is part of an
      alias base.

      \begin{definition}[Alias agreement]\label{def:AliasAgree}
	We say that
	$$(env,(w,\sqleq_w),\loc)\models(\Gamma,\kappa)$$
	if
	\begin{itemize}
		\item $\exists \loc^p\in \dom(w).\nu x^p\in \dom(\Gamma)\Rightarrow\nu x\in\kappa$
		\item $env^{-1}(\loc)\neq\emptyset.\exists \kappa^0_i\in\kappa^0\Rightarrow
			(env^{-1}(\loc)\subseteq\kappa^0_i)\wedge(\exists \loc^p\in \dom(w).\nu x^p\in \dom(\Gamma)\Rightarrow\nu x\in\kappa^0_i\wedge\nu x\in\kappa)$
		\item $env^{-1}(\loc)=\emptyset.\exists \kappa^0_i\in\kappa^0\Rightarrow
			(\exists \loc^p\in \dom(w).\nu x^p\in \dom(\Gamma)\Rightarrow\nu x\in\kappa^0_i\wedge\nu x\in\kappa)$
	\end{itemize}
\end{definition}

If a value $v$ is a location, then we check that both the set of occurrences agrees with the dependency pair, presented in \cref{def:DepAgree}, 
and check if the alias information agrees with the semantics, \cref{def:AliasAgree}.
If the value $v$ is not a location, then its type can either be an
abstraction type or a base type.
For the base type, we check that the agreement between the set of occurrences and the dependency pair agrees.
If the type is an abstraction, then we check that $T_2$ agrees with 
the binding model. 
We are only concerned about the return type $T_2$ for abstractions,
since if the argument parameter is used in the body of the
abstraction, then the dependencies would already be part of the return
type. 

\begin{definition}[Type agreement]\label{def:TAgree}
	We say that
	$$(env,v,(w,\sqleq_w),(L,V))\models(\Gamma,T)$$
	iff
	\begin{itemize}
		\item $v\neq\loc$ and $T=T_1\rightarrow T_2$:
		\begin{itemize}
			\item $(env,v,(w,\sqleq_w),(L,V))\models(\Gamma,T_2)$
		\end{itemize}

		\item $v\neq\loc$ and $T=(\delta,\kappa)$:
		\begin{itemize}
			\item $(env,(L,V))\models\delta$
		\end{itemize}

		\item $v=\loc$ then $T=(\delta,\kappa)$ where:
		\begin{itemize}
			\item $(env,(L,V))\models\delta$
			\item $(env,(w,\sqleq_w),v)\models(\Gamma,\kappa)$
		\end{itemize}
	\end{itemize}
\end{definition}


\begin{definition}[Environment agreement]\label{def:EnvAgree}
	We say that $(env,sto,(w,\sqleq_w))\models(\Gamma,\Pi)$ if 
	\begin{enumerate}
		\item \label{prop:1} $\forall x\in \dom(env).(\exists x^p\in \dom(w))\wedge(x^p\in \dom(w)\Rightarrow \exists x^p\in \dom(\Gamma))$
		\item \label{prop:2} $\forall x^p\in \dom(w).x^p\in \dom(\Gamma)\Rightarrow env(x)=v\wedge w(x^p)=(L,V)\wedge\Gamma(x^p)=T.\\(env,v,(w,\sqleq_w),(L,V))\models (\Gamma,T)$
		\item \label{prop:3} $\forall \loc\in \dom(sto).(\exists \loc^p\in \dom(w))\wedge(\exists \nu x.\forall p\in\{p'\mid\loc^{p'}\in \dom(w)\}\Rightarrow\\\nu x^p\in \dom(\Gamma))$
		\item \label{prop:4} $\forall \loc^p \in \dom(w).\exists\nu x^{p}\in \dom(\Gamma)\Rightarrow w(\loc^p)=(L,V)\wedge\Gamma(\nu x^{p})=\\T.(env,\loc,(w,\sqleq_w),(L,V))\models T$
		\item \label{prop:5} if $p_1\sqleq_w p_2$ then $p_1\sqleq_\Pi p_2$
	\item \label{prop:6} $\forall \loc^p\in \dom(w).\exists \nu
          x^p\in \dom(\Gamma)\Rightarrow\exists
          p'\in\cat{P}.\IP_{\sqleq_w}(\loc,w)\in
          \IP_{\Upsilon_{p'}}(\nu x,\Gamma)$ 
	\end{enumerate}
\end{definition}



\begin{itemize}
	\item The agreement for local information only relates the information currently known by $env$, and that the information known by $w$ and $\Gamma$ agrees, in respect to \cref{def:TAgree}.
		This is ensured by \eqref{prop:1} and \eqref{prop:2}.

	\item We similarly handle agreement for the global information
          known, which is ensured by \eqref{prop:3} and
          \eqref{prop:4}. 
		Since $\Gamma$ contains the global information for
                references, we require that there exists a
                corresponding internal variable to the currently known
                locations, by comparing them by program points. 
		We also ensure that the dependency information for a location
                occurrence agrees with the type of a corresponding
                internal variable occurrence as given by \cref{def:TAgree}. 

	\item We also need to ensure that $\Pi$ is a good
          approximation of the order $\sqleq_w$ and the greatest
          binding function for $p$-chains ensures that we always get
          the necessary reference occurrences. 
	\eqref{prop:5} ensures that the ordering information
        $\sqleq_w$ agrees with that of $\Pi$.

   \item We finally need to ensure that for every location
        known, there exists a corresponding internal variable
        where, getting the greatest binding of this occurrence,
        $\loc^p$, there exists a program point $p'$, such that
        looking up all greatest bindings for the $p'$-chain, there
        exists an internal variable occurrence that corresponds to
        $\loc^p$. This is captured by \eqref{prop:6}.
\end{itemize}




\begin{lemma}[History]\label{lemma:His}
	Suppose $e^p$ is an occurrence, that
	$$env\vdash\left\langle e^{p},sto,(w,\sqleq_w),p'\right\rangle\rightarrow\left\langle v,sto',(w',\sqleq_w'),(L,V),p''\right\rangle$$
		and $x^{p_1}\in \dom(w')\backslash \dom(w)$.
		Then $x\notin fv(e^{p})$
\end{lemma}



\begin{lemma}[Strengthening]\label{lemma:Strength}
	If $\Gamma[x^{p'}:T'],\Pi\vdash e^{p}:T$ and $x\notin fv(e^p)$, then $\Gamma,\Pi\vdash e^{p}:T$
\end{lemma}

\subsection{The soundness theorem}

We can now present the soundness theorem for our type system.

\begin{theorem}[Soundness]
	Suppose $e^{p'}$ is an occurrence where
	\begin{itemize}
		\item $env\vdash\left\langle e^{p'},sto,(w,\sqleq_w),p\right\rangle\rightarrow\left\langle v,sto',(w',\sqleq_w'),(L,V),p''\right\rangle$,
		\item $\Gamma,\Pi\vdash e^{p'} : T$
		\item $\Gamma,\Pi\vdash env$
		\item $(env,sto,(w,\sqleq_w))\models(\Gamma,\Pi)$
	\end{itemize}
	Then we have that
	\begin{itemize}
		\item $\Gamma,\Pi\vdash v:T$
		\item $(env,sto',(w',\sqleq_w'))\models(\Gamma,\Pi)$
		\item $(env,(w',\sqleq_w'),v,(L,V))\models(\Gamma,T)$
	\end{itemize}
\end{theorem}
\begin{proof}(Outline)
	The proof proceeds by induction on the height of the derivation tree for 
	$$env\vdash\left\langle e^{p'},sto,\psi,p\right\rangle\rightarrow\left\langle v,sto',\psi',(L,V),p''\right\rangle$$
	We will only show the proof of four rules here, for
        \runa{Var}, \runa{Case}, \runa{Ref}, and \runa{Ref-write}.

	\begin{description}
		\input{sections/Proof/SoundProof/var.tex}
		\input{sections/Proof/SoundProof/case.tex}
		\input{sections/Proof/SoundProof/refread.tex}
	\end{description}
\end{proof}


\section{An implementation}

We have made an implementation in the Rust programming language. It
includes a parser, and evaluator, a type checker, and an approximator
for an order of programs points. The approximator is based on the
given type system, where it is derived from the structure of the type
system.  The implementation is hosted on Github and can be found at
\cite{implementation}.

\section{Conclusion}\label{sec:Conc}

We have introduced a type system for local data-flow analysis for a
subset of ReScript that includes functional as well as imperative
feature, notably that of references.

The type system provides a safe approximation of the data flow in an
expression. This also allows us to reason about security
properties. In particular, the notion of non-interference introduced
by Goguen and Meseguer \cite{goguen-meseguer} and studied in information-flow analysis can be
understood in this setting. A program satisfies the non-interference
property if the variables classified as \emph{low} cannot be affected
by variables classified as \emph{high}. This corresponds to the
absence of chains in $\Pi$ in which low occurrences appear below high
occurrences. A topic of further work is to understand the relative
expressive power of our system wrt. the systems of Volpano and Smith
\cite{volpano-smith-96,volpano-smith-97}.

On the other hand, the system contains slack. In particular, the type
system is monomorphic. This means that a locally declared abstraction
cannot be used at multiple places, even though this may be safe, as
this would mean it would contain occurrences at multiple program
points. Moreover, abstractions cannot be bound to references.

Polymorphism for the base type $(\delta,\kappa)$ would allow 
abstractions to be used multiple times in an occurrence. Consider as
an example

\begin{lstlisting}[language=Caml, mathescape=true]
(let x ($\lambda$ y.y$^1$)$^2$ (x$^3$ (x$^4$ 1$^5$)$^6$)$^7$)$^8$
\end{lstlisting}

Occurrences such as this would now become typable,
since when typing the applications, the type of the argument changes,
as the occurrence $x^4$ is present in the second application. 

The way references are defined currently in the type system, they cannot be bound to abstractions.
If this should be introduced a couple of questions need to be evaluated.
First there should be looked into base type polymorphism, the second
would be to look into type polymorphism, i.e., allow a reference to be bound to
an arrow type at one point and a base type at another.




A next step is to devise a type inference algorithm for the type
system. An inference algorithm must compute an approximated order of
program points, a proper $\kappa_0$ and the types for abstractions,
that is, find all the places where the parameter of an abstraction
should be bound.  We conjecture that such a type inference algorithm
for our system will be able to compute the information found in an
interative data flow analysis.

\end{document}

%% file: setup/preamble-from2024.tex
\usepackage{amsmath,amssymb}
\usepackage{balance}

\usepackage[capitalise]{cleveref}

\usepackage{float}
\usepackage{listings}
\usepackage{placeins}
\usepackage{subcaption}
\usepackage{wrapfig}
\usepackage{booktabs}
\usepackage{tabularx}
\usepackage[most]{tcolorbox}

\usepackage{filecontents}
\usepackage{stmaryrd}

\usepackage{xspace}
\usepackage{todonotes}
\usepackage{subfiles}

\usepackage{algpseudocode}
\usepackage{algorithm}
\usepackage[acronym]{glossaries}
\usepackage[T1]{fontenc}
\usepackage{semantic}

%% file: setup/macros-from2024.tex
\newcommand{\runa}[1]{\textsc{(#1)}}
\newcommand{\cat}[1]{\textbf{#1}}

\newcommand{\loc}{\ensuremath{\ell}}

\newcommand{\Pow}[1]{\ensuremath{\mathbb{P}(#1)}}

\algnewcommand\algorithmicforeach{\textbf{for each}}
\algdef{S}[FOR]{ForEach}[1]{\algorithmicforeach\ #1\ \algorithmicdo}

\newtheorem{theorem}{Theorem}
\newtheorem{lemma}{Lemma}

\theoremstyle{definition}
\newtheorem{definition}{Definition}
\newtheorem{example}{Example}

\newcommand{\dom}{\textrm{dom}}

\newcommand{\condinf}[3]{\begin{tabular}{l} \inference[]{#1}{#2} \\[5mm]
                          #3 \end{tabular}}

\newcommand{\sqleq}{\ensuremath{\sqsubset\xspace}}

\newcommand{\refc}{\ensuremath{\textrm{ref}\xspace}}
\newcommand{\points}{\ensuremath{\textsf{points}}}
\newcommand{\nexte}{\ensuremath{next}\xspace}

\newcommand{\IP}{\ensuremath{\textsf{IP}}\xspace}

\lstnewenvironment{rescript}[1][]
      {\lstset{language=Caml}\lstset{escapeinside={(*@}{@*)},
       breaklines=true,
       mathescape=true,
           framesep=5pt,
           showstringspaces=false,
           keywordstyle=\bfseries\color{black},
           stringstyle=\color{maroon},
           commentstyle=\color{black},
           moredelim=**[is][\only<+>{\color{red}}]{@}{@},
           gobble=4,
        rulecolor=\color{black},
        xleftmargin=0pt,
        xrightmargin=0pt,
        aboveskip=\bigskipamount,
        belowskip=\bigskipamount,
               backgroundcolor=\color{white}, #1
    }}
    {}



%% file: sections/Proof/SoundProof/var.tex
\item[\runa{Var}] Here $e^{p'}=x^{p'}$, where
\begin{figure}[H]
	\setlength\tabcolsep{8pt}
	\begin{tabular}{l}
		\input{sections/appendix/ColRules/var.tex}
	\end{tabular}
\end{figure}
And from our assumptions, we have:
\begin{itemize}
	\item $\Gamma,\Pi\vdash x^{p'} : T$
	\item $\Gamma,\Pi\vdash env$
	\item $(env,sto,(w,\sqsubseteq_w))\models(\Gamma,\Pi)$
\end{itemize}
To type the occurrence $x^{p'}$ we use the rule \runa{T-Var}:
\begin{figure}[H]
	\setlength\tabcolsep{8pt}
	\begin{tabular}{l}
		\runa{T-Var}\\[0.2cm]
			\inference[]{}
			{\Gamma,\Pi \vdash x^p:T \sqcup (\{x^p\},\emptyset)}
	\end{tabular}
\end{figure}
Where $x^{p''}=uf_{\sqsubseteq_\Pi}(x,\Gamma)$, $\Gamma(x^{p''})=T$.

We need to show that \cat{1)} $\Gamma,\Pi\vdash c:T$, \cat{2)} $(env,sto',(w',\sqsubseteq_w'))\models(\Gamma,\Pi)$, and\\
\cat{3)} $(env,v,(w',\sqsubseteq_w'),(L,V))\models(\Gamma,T)$.
\begin{description}
	\item[1)] Since, from our assumption, we know that $\Gamma,\Pi\vdash env$, we can then conclude that $\Gamma,\Pi\vdash v:T$

	\item[2)] Since there are no updates to $sto$ and $(w,\sqsubseteq_w)$, we then know that $(env,sto,(w,\sqsubseteq_w))\models(\Gamma,\Pi)$ holds after an evaluation.

	\item[3)] Since there are no updates to $sto'$ and
          $(w',\sqsubseteq_w')$, since $(L,V) = w(x^{p''})$ and since
          $T = \Gamma(x^{p''})$, we then know that
          $(env,v,(w',\sqsubseteq_w'),(L,V))\models(\Gamma,T)$.  Due
          to \cref{def:TAgree} we can conclude that:
		$$(env,v,(w',\sqsubseteq_w'),(L,V\cup\{x^{p''}\}))\models(\Gamma,T\sqcup \{x^{p''}\})$$
\end{description}

%% file: sections/appendix/ColRules/var.tex
\runa{Var}\\[0.2cm]
	\inference[]{}
	{env\vdash \left\langle x^{p'},sto,(w,\sqsubseteq_w),p \right\rangle \rightarrow \left\langle v,sto,(w,\sqsubseteq_w),(L,V\cup\{x^{p'}\}),p' \right\rangle}\\[0.3cm]
	Where $env(x)=v$, $x^{p''}=uf_{\sqsubseteq_w}(x,w)$, and $w(x^{p''})=(L,V)$

%% file: sections/Proof/SoundProof/case.tex
\item[\runa{Case}] Here $e^{p'}=\left[\mbox{case}\;e^{p''}\;\tilde{\pi}\;\tilde{o}\right]^{p'}$, where
\begin{figure}[H]
	\setlength\tabcolsep{8pt}
	\begin{tabular}{l}
		\input{sections/appendix/ColRules/case.tex}
	\end{tabular}
\end{figure}

And from our assumptions, we have that:
\begin{itemize}
	\item $\Gamma,\Pi\vdash \left[\mbox{case}\;e^{p''}\;\tilde{\pi}\;\tilde{o}\right]^{p'}:T$,
	\item $\Gamma,\Pi\vdash env$
	\item $(env,sto,(w,\sqsubseteq_w))\models(\Gamma,\Pi)$,
\end{itemize}
To type $[\mbox{case}\;e^{p''}\;\tilde{\pi}\;\tilde{o}]^{p'}$ we need to use the \runa{T-Case} rule, where we have:
\begin{figure}[H]
	\setlength\tabcolsep{8pt}
	\begin{tabular}{l}
		\runa{T-Case}\\[0.2cm]
			\inference[]
				{\Gamma,\Pi\vdash e^{p}:(\delta,\kappa) &\\
				\Gamma',\Pi\vdash e_i^{p_i}:T_i\;\;\;(1\leq i\leq|\tilde{\pi}|)}
				{\Gamma,\Pi\vdash [\mbox{case}\;e^{p}\;\tilde{\pi}\;\tilde{o}]^{p'}:T}
	\end{tabular}
\end{figure}
Where $T=T'\sqcup(\delta,\kappa)$, $T'=\bigcup_{i=1}^{|\tilde{\pi}|}T_i$, $e_i^{p_i}\in\tilde{o}$ and $s_i\in\tilde{\pi}$, and $\Gamma'=\Gamma[x^p:(\delta,\kappa)]$ if $s_i=x$.

We must show that \cat{1)} $\Gamma,\Pi\vdash v:T$, \cat{2)} $(env,sto',(w',\sqsubseteq_w'))\models(\Gamma,\Pi)$, and \\
\cat{3)} $(env,v,(w',\sqsubseteq_w'),(L,V))\models(\Gamma,T)$.

To conclude, we first need to show for the premises, where due to our assumption and from the first premise, we can use the induction hypothesis to get:
\begin{itemize}
	\item $\Gamma,\Pi\vdash v_e:(\delta,\kappa)$,
	\item $(env,sto'',(w'',\sqsubseteq_w''))\models(\Gamma,\Pi)$,
	\item $(env,v,(w'',\sqsubseteq_w''),(L,V))\models(\Gamma,(\delta,\kappa))$
\end{itemize}
Since in the rule \runa{T-Case} we take the union of all patterns, we can then from the second premise:
\begin{itemize}
	\item $\Gamma,\Pi\vdash v:T_j$,
	\item $(env,sto',(w',\sqsubseteq_w'))\models(\Gamma,\Pi)$,
	\item $(env,v,(w',\sqsubseteq_w'),(L,V))\models(\Gamma,T_j)$
\end{itemize}

If we have \cat{a)} $\Gamma',\Pi\vdash env[env']$ and \cat{b)} $(env[env'],sto'',(w''',\sqsubset_w''))\models(\Gamma',\Pi)$, we can then conclude the second premise by our induction hypothesis.
\begin{description}
	\item[a)] We know that either we have $\Gamma'=\Gamma[x\mapsto(\delta,\kappa)]$ and $env[x\mapsto v_e]$ if $s_j=x$, or $\Gamma'=\Gamma$ and $env$ if $s_j\neq x$.
		\begin{itemize}
			\item if $s_j\neq x$: Then we have $\Gamma,\Pi\vdash env$
			\item if $s_j=x$: Then we have $\Gamma[x\mapsto(\delta,\kappa)],\Pi\vdash env[x\mapsto v_e]$, which hold due to the first premise.
		\end{itemize}
	\item[b)] We know that either we have $\Gamma'=\Gamma[x\mapsto(\delta,\kappa)]$ and $env[x\mapsto v_e]$ if $s_j=x$, or $\Gamma'=\Gamma$ and $env$ if $s_j\neq x$.
		\begin{itemize}
			\item if $s_j\neq x$: then we have $(env,sto'',(w'',\sqsubset_w''))\models(\Gamma,\Pi)$.
			\item if $s_j=x$: then $(env[x\mapsto v_e],sto'',(w''',\sqsubset_w''))\models(\Gamma[x\mapsto(\delta,\kappa)],\Pi)$, since we know that $(env,sto'',(w'',\sqsubset_w''))\models(\Gamma,\Pi)$, we only need to show for $x$.
				Since we have $x\in dom(env)$, $x^{p_j}\in dom(w''')$ and $x^{p_j}\in dom(\Gamma')$ and due to the first premise, we know that $(env[x\mapsto v_e],sto'',(w''',\sqsubset_w''))\models(\Gamma[x\mapsto(\delta,\kappa)],\Pi)$.
		\end{itemize}
\end{description}
Based on \cat{a)} and \cat{b)} we can then conclude:

\begin{description}
	\item[1)] Since $\Gamma',\Pi\vdash v:T_j$, then we also must have $\Gamma',\Pi\vdash v:T$, since $T$ only contains more information than $T_j$.
	\item[2)] By the second premise, \cref{lemma:His}, and \cref{lemma:Strength}, we can then get 
		$$(env,sto',(w',\sqsubseteq_w'))\models(\Gamma,\Pi)$$
	\item[3)] Due to \cat{1)}, \cat{2)}, \cat{a)}, and \cat{b)} we can then conclude that
		$$(env,v,(w',\sqsubseteq_w'),(L,V))\models(\Gamma,T)$$
\end{description}

%% file: sections/appendix/ColRules/case.tex
\runa{Case}\\[0.2cm]
	\inference[]
	{
		env \vdash \left\langle e^{p''},sto,(w,\sqsubseteq_w),p \right\rangle \rightarrow \left\langle v_e,sto'',(w'',\sqsubseteq_w''),(L'',V''),p'' \right\rangle &\\
		env[env'] \vdash \left\langle e_j^{p_j},sto'',(w''',\sqsubseteq_w''),p'' \right\rangle \rightarrow \left\langle v,sto',(w',\sqsubseteq_w'),(L',V'),p_i \right\rangle
	}
	{env\vdash \left\langle \left[\mbox{case}\;e^{p''}\;\tilde{\pi}\;\tilde{o}\right]^{p'},sto,(w,\sqsubseteq_w),p \right\rangle \rightarrow \left\langle v,sto',(w',\sqsubseteq_w'),(L,V),p' \right\rangle}\\[0.3cm]
	Where $match(v_e,s_i)=\perp$ for all $1\leq u<j\leq|\tilde{\pi}|$, $match(v_e,s_j)=env'$, and \\
	$w'''=w''[x\mapsto(L'',V'')]$ if $env'=[x\mapsto v_e]$ else $w'''=w''$

%% file: sections/Proof/SoundProof/refread.tex
\item[\runa{Ref-read}] Here $e^{p'}=[!e_1^{p_1}]^{p'}$, where
\begin{figure}[H]
	\setlength\tabcolsep{8pt}
	\begin{tabular}{l}
		\input{sections/appendix/ColRules/refread.tex}
	\end{tabular}
\end{figure}
And from our assumptions, we have that:
\begin{itemize}
	\item $\Gamma,\Pi\vdash [!e_1^{p_1}]^{p'}:T$,
	\item $\Gamma;\Pi\vdash env$
	\item $(env,sto,(w,\sqsubseteq_w))\models(\Gamma,\Pi)$,
\end{itemize}
To type $[!e_1^{p_1}]^{p'}$ we need to use the \runa{T-Ref-read} rule, where we have:
\begin{figure}[H]
	\setlength\tabcolsep{8pt}
	\begin{tabular}{l}
		\runa{T-Ref-read}\\[0.2cm]
			\inference[]
				{\Gamma,\Pi\vdash  e^{p}:(\delta,\kappa)}
				{\Gamma,\Pi\vdash [!e^{p}]^{p'}:T\sqcup(\delta\cup\delta',\emptyset)}\\
	\end{tabular}
\end{figure}
Where $\kappa\neq\emptyset$, $\delta'=\{\nu x^{p'}\mid\nu x\in\kappa\}$, $\nu x_1,\cdots,\nu x_n\in\kappa$.\\ 
$\{\nu x_1^{p_1},\cdots,\nu x_1^{p_m}\}=uf_{\Upsilon_{p'}}(\nu x_1,\Gamma),\cdots,\{\nu x_n^{p_1'},\cdots,\nu x_n^{p_s'}\}=uf_{\Upsilon_{p'}}(\nu x_n,\Gamma)$, and\\
$T=\Gamma(\nu x_1^{p_1})\cup\cdots\cup\Gamma(\nu x_1^{p_m})\cup\cdots\cup\Gamma(\nu x_n^{p_1'})\cup\cdots\cup\Gamma(\nu x_n^{p_s'})$.

We must show that \cat{(1)} $\Gamma,\Pi\vdash v:T$, \cat{(2)} $(env,sto',(w',\sqsubseteq_w'))\models(\Gamma,\Pi)$, and\\
\cat{(3)} $(env,v,(w',\sqsubseteq_w'),(L,V))\models(\Gamma,T)$.

To conclude, we first need to show for the premises, where due to our assumption and from the premise, we can use the induction hypothesis to get:
\begin{itemize}
	\item $\Gamma,\Pi\vdash \loc:(\delta,\kappa)$,
	\item $(env,sto',(w',\sqsubseteq_w'))\models(\Gamma,\Pi)$,
	\item $(env,v,(w',\sqsubseteq_w'),(L,V))\models(\Gamma,(\delta',\kappa'))$
\end{itemize}

Due to $(env,sto',(w',\sqsubseteq_w'))\models(\Gamma,\Pi)$ and $(env,v,(w',\sqsubseteq_w'),(L,V))\models(\Gamma,(\delta',\kappa'))$, and due to our assumptions, we can conclude that:
\begin{description}
	\item[(1)] $\Gamma,\Pi\vdash v:T$,

	\item[(2)] $(env,sto',(w',\sqsubseteq_w'))\models(\Gamma,\Pi)$,

	\item[(3)] $(env,v,(w',\sqsubseteq_w'),(L\cup\{\loc^{p''}\},V))\models(\Gamma,T\sqcup(\delta\cup\delta',\emptyset))$
\end{description}

%% file: sections/appendix/ColRules/refread.tex
\runa{Ref-read}\\[0.2cm]
	\inference[]
	{env \vdash \left\langle e^{p_1},sto,(w,\sqsubseteq_w),p \right\rangle \rightarrow \left\langle \loc,sto',(w',\sqsubseteq_w'),(L_1,V_1),p_1 \right\rangle}
	{env\vdash \left\langle \left[!e^{p_1}\right]^{p'},sto,(w,\sqsubseteq_w),p \right\rangle \rightarrow \left\langle v,sto',(w',\sqsubseteq_w'),(L\cup L_1\cup\{\loc^{p''}\},V\cup V_1),p' \right\rangle}\\[0.3cm]
	Where $sto'(\loc)=v$, $\loc^{p''}=uf_{\sqsubseteq_w'}(\loc,w')$, and $w'(\loc^{p''})=(L,V)$